\documentclass[11pt]{article}
\usepackage{geometry}
\usepackage{amssymb}
\usepackage{mathrsfs}
\usepackage{amsmath}

\newtheorem{theorem}{Theorem}[section]

\newtheorem{lemma}[theorem]{Lemma}

\newtheorem{definition}[theorem]{Definition}

\def\squarebox#1{\hbox to #1{\hfill\vbox to #1{\vfill}}}
\newcommand{\qed}{\hspace*{\fill}
\vbox{\hrule\hbox{\vrule\squarebox{.667em}\vrule}\hrule}\smallskip}

\newenvironment{proof}{\noindent{\bf Proof:~~}}{\(\qed\)}

\newcommand{\ignore}[1]{}

\begin{document}

\title{A Note on the Power of Truthful Approximation Mechanisms}

\author{Shahar Dobzinski   \thanks{Supported by the Adams Fellowship
   Program of the Israel Academy of Sciences and Humanities, and by a
   grant from the Israeli Academy of
   Sciences. Email: {\tt shahard@cs.huji.ac.il}.}}

  \maketitle
\begin{abstract}
We study the power of polynomial-time truthful mechanisms comparing to polynomial time (non-truthful) algorithms. We show that there is a setting in which deterministic polynomial-time truthful mechanisms cannot guarantee a bounded approximation ratio, but a non-truthful FPTAS exists. We also show that in the same setting there is a universally truthful randomized mechanism that provides an approximation ratio of $2$. This shows that the cost of truthfulness is unbounded. The proofs are almost standard in the field and follow from known results.
\end{abstract}

\section{Introduction}

A large body of the literature in algorithmic mechanism design deals with understanding the power of polynomial time truthful mechanisms comparing to the power of polynomial time non-truthful algorithms. In other words, given a polynomial time approximation algorithm that provides an approximation ratio of $ \alpha$, determine if there is a polynomial time truthful mechanism that provides at least the same approximation ratio. Lavi, Mu'alem, and Nisan \cite{LMN03} gave the first negative answer to this question: there exists a setting that admits an FPTAS ignoring incentives, but truthful polynomial time mechanisms cannot provide an approximation ratio better than 2.

Recently, some papers attempted to determine the size of the gap. That is, is there a setting where a poly time (non-truthful) algorithm provides an approximation ratio of $ \alpha$, but no truthful algorithm provides an approximation ratio of $ c\cdot \alpha$, where $ c$ is very large? Papadimitriou, Singer, and Schapira \cite{PSS08} presented the problem of combinatorial public projects and showed that for this problem there exists an approximation algorithm that provides a constant approximation ratio, while no truthful algorithm can provide an approximation ratio better than $ \sqrt m$, thus showing that $ c=\Omega(\sqrt m)$. In this note we will completely answer this question and show a setting where there is an FPTAS, but poly time truthful mechanisms cannot guarantee any bounded approximation ratio. The setting is an artificial one designed for the sole purpose of proving the result. Providing a ''natural'' setting that proves such a separation remains a question of some interest. We also show that poly time universally truthful randomized mechanisms provide an approximation ratio of $2$ for this setting, thus also establishing that the gap between the power of randomized poly-time mechanisms and deterministic ones is unbounded too. We note that although our results are communication-complexity based, establishing computational-complexity hardness can be done very similarly.

\vspace{0.1in}\noindent \textbf{The Setting}: in an \emph{extended multi-unit auction} we have $ m$ identical items and $ 2$ bidders. Each bidder $ i$ has a valuation function $ v_i$ that gives a value for each number of the items. Each $ v_i$ is almost monotone: for each bidder $ i$ there exists some $ k_i$, $ 0<k_i \leq m$, such that $v(m)\geq v(m-1)\geq ...\geq v(k_i+1)\geq v(k_i-1)\geq v(k_i-2)\geq ...\geq v(1)$. The valuations are normalized: $v(0)=0$. The goal is to maximize the welfare given allocations that always allocate all items: find an allocation $(t,m-t)$ that maximizes $ v(t)+v(m-t)$.

\vspace{0.1in}\noindent The private information of each bidder is $ v_i$ (notice that the $ k_i$'s are private information too). We will assume that the $ v_i$'s are given to us as black boxes, and will be interested in algorithms that run in time polynomial in $ \log m$ (why? The short answer is that we have a (non-truthful) algorithm that provides a $ (1+\epsilon)$-approximation in time $ poly(\log m, \frac 1 \epsilon)$, so the running time of our truthful mechanisms should be polynomially related. See Dobzinski and Nisan \cite{DN07b} for further motivation.)

\section{Preliminaries: Affine Maximizers and Maximal-in-Range Algorithms}

Arguably the main positive result of mechanism design is the VCG payment scheme. In this scheme we obtain an optimal solution $ (o_1,\ldots,o_n)$, and allocate accordingly. Each bidder $ i$ is being paid $ \Sigma_{j\neq i} v_i(o_i)$. Thus, the profit of each bidder $ i$ equals to the value of the optimal solution: $ v_i(o_i)+\Sigma_{j\neq i} v_i(o_i)$. Observe that truthfulness is a dominant strategy for each bidder: by deviating the bidder might enforce a suboptimal allocation, his profit will then be equal to the welfare of this suboptimal profit, which results in a decreased profit. What about the running time? At a first glance, it looks that VCG does not help much in constructing polynomial-time mechanisms, as it requires finding the optimal solution, and finding the optimal solution is computationally hard. One naive solution might be to use an approximation algorithm that finds an approximate solution $ (s_1,\ldots,s_n)$, and uses VCG payments with respect to this approximate solution. I.e., the mechanism pays each bidder $ i$ $ \Sigma_{j\neq i} v_i(s_i)$. Unfortunately, in general, using VCG payments together with an arbitrary approximation algorithm does not result in a truthful mechanism. There are, however, some approximation algorithms that are truthful when used with the VCG payment scheme:

\begin{definition}
An algorithm is called maximal in range (MIR) if there exists a subset $ \mathcal R$ of allocations (the range of the algorithm), such that for every possible input $ v_1,\ldots, v_n$, the algorithm outputs the allocation that maximizes the welfare in $ \mathcal R$. I.e., for all input valuations $ v_1,\ldots,v_n$ the algorithm outputs $ \arg\max _{(S_1,\ldots,S_n)\in\mathcal R}\Sigma_iv_i(S_i)$.
\end{definition}

Of course, we can also implement a weighted versions of VCG:

\begin{definition}
An algorithm is an affine maximizer if there exists a subset $ \mathcal R$ of allocations, a constant $ w_i$ for each bidder $ i$, and a constant $ c_R$ for each $ R\in \mathcal R$ such that for all input valuations $ v_1,\ldots,v_n$ the algorithm outputs $ \arg\max _{(S_1,\ldots,S_n)\in\mathcal R}\Sigma_iw_iv_i(S_i)+C_R$.
\end{definition}

We will show that there are settings in which all truthful mechanisms are affine maximizers.

\section{The Power of Deterministic Poly-Time Mechanisms vs Poly-Time Algorithms}

In this section we deal only with deterministic mechanisms. In the second section we will consider randomized ones. The proof is similar to the proof of Dobzinski and Nisan \cite{DN07b} for the setting of multi-unit auctions. We will show that all truthful algorithms for extended multi-unit auctions are affine maximizers. Then we will show that affine maximizers cannot provide a finite approximation ratio in polynomial time. The last part follows easily from \cite{DN07b}. This will conclude the proof.

The proof is slightly novel from a technical point of view: previous proofs of gaps for other settings obtain a characterization "from scratch" using Roberts-like machinery. In contrast, we obtain a characterization by a reduction to a known characterization result.

We require some notation. Let $ D$ be the set of all valuations of extended multi-unit auctions. Let $ D'$ be the set of all monotone and normalized valuations ($ v(0)=0$, and $ v(m)\geq v(m-1)\geq ...\geq v(1)\geq v(0)$).

\begin{lemma}
Let $ f$ be a truthful mechanism for extended multi-unit auctions with range at least $3$. $ f$ is an affine maximizer.
\end{lemma}

\begin{proof}
Consider restricting the domain of $ f$ to valuations in $ D'$ only. We claim that the in this case $ f$ implements an affine maximizer. I.e., there exists a subset of all allocations $ \mathcal R$, constants $ w_1,w_2$ and a constant $ c_R$ for each allocation $ R\in\mathcal R$ such that $ f(v_1,v_2)=\arg\max_{R\in \mathcal R} w_1v(1)+w_2v(2)+c_R$, for all valuations $ v_1,v_2\in D$. This follows from a paper with Sundararajan \cite{DS08-add}, as if both valuations are from $ D'$ this setting is equivalent to multi-unit auctions where all items are allocated.

The taxation principle says that in a truthful mechanism every bidder is faced with a price for each alternative (that depends only on the valuations of the other bidders), and that this bidder chooses its most profitable alternative. Since these prices depend only on the valuation of the other bidder we will say that these prices are induced by the other bidder. If both valuations are in $ D'$ then each bidder induces prices that are the same as weighted VCG (since in this case the implemented function is an affine maximizer). Since the prices that a bidder induces do not depend on the valuation of the other bidder, we have that if a bidder's valuation is in $ D$ then the prices he induces are the prices of weighted VCG (even if the other bidder valuation is in $ D'$).

Thus, if $ v_1\in D$ and $ v_2\in D'$, bidder $ 2$ will induce weighted VCG prices and bidder $ 1$ will choose the allocation that maximizes the weighted welfare. Hence, $ f(v_1,v_2)$ is an affine maximizer if $ v_1\in D$ and $ v_2\in D'$. Similarly, since the prices that a bidder induces depend only on his own valuation, we have that bidder $ 2$ always induces weighted VCG prices. We can now conclude that $ f$ is an affine maximizer. This concludes the proof of the lemma.
\end{proof}

\begin{lemma}
Let $ f$ be an affine maximizer for extended multi-unit auctions that provides bounded approximation ratio. The communication complexity of $ f$ is $ m$.
\end{lemma}

\begin{proof}
We will prove this theorem only for maximal-in-range algorithms. Extending the proof to affine maximizers is easy. We first claim that if $ f$ provides a bounded approximation ratio then the range of $ f$ must be full. Suppose not, and let the allocation $ (t,m-t)$ be the allocation that is not in the range. Now consider the following instance: $ v_1$ and $ v_2$ are identically zero except for $ v_1(t)=v_2(m-t)=\infty$. We have that $ f$ does not provide a bounded a bounded approximation ratio, in contradiction to our assumption.

Thus, $ f$ is a maximal-in-range algorithm that has a full range, i.e., an algorithm that always outputs the optimal solution. Nisan and Segal \cite{NS06} show that optimally solving multi-unit auctions requires a communication complexity of $ m$, and their proof directly extends to case of extended multi-unit auctions. Thus the communication complexity of $ f$ is $ m$, as needed.
\end{proof}

All we have to show now is that there is an FPTAS for the problem. Consider the following algorithm: each bidder $ i$ rounds down each value to the nearest power of $ (1+\epsilon)$, and reports only the (logarithmic number of) bundles where his (rounded down) valuation changed. Each bidder also reports $ v_i(k_i)$. Now we use exhaustive search to find the optimal allocation that consists only of bundles that their value was reported. The algorithm finds a solution that recovers a fraction of at least $ 1-\frac 1 {(1+\epsilon)}$ of the optimal welfare.

\begin{theorem}
There exists a setting in which there is a (non-truthful) FPTAS, but truthful mechanisms cannot guarantee any bounded approximation ratio.
\end{theorem}

\section{The Power of Deterministic Poly-Time Mechanisms vs. Randomized ones}

Another question is to determine the power of randomization. The two main types of randomized mechanisms that are considered in algorithmic mechanism design are:

\begin{itemize}
    \item \textbf{Universally truthful mechanisms: }these mechanisms are a probability distribution over deterministic mechanisms.
    \item \textbf{Truthful in Expectation mechanisms: }the mechanism flips some random coins. The dominant strategy for a bidder that is interested in maximizing the \emph{expected} profit is to reveal his true valuation.
\end{itemize}

Notice that every universally truthful mechanism is also truthful in expectation.

The following is a poly time universally truthful mechanism that provides an approximation ratio of $ 2$ for extended multi-unit auctions: select a bidder $i$ uniformly at random and allocate him $ \arg\max (v_i(m),v_i(k_i))$ items. We leave the proof of the approximation ratio and truthfulness as an easy exercise to the reader. This proves that the gap between the power of poly time universally truthful and poly time deterministic truthful mechanisms is unbounded.

In a recent work with Dughmi \cite{DD09}, we show for the first time that poly time truthful-in-expectation mechanisms are strictly more powerful than universally truthful ones: there is a setting where universally truthful polynomial time mechanisms cannot provide an approximation ratio better than $ 2$, but there exists a truthful in expectation FPTAS. Increasing this gap remains an open question.

\bibliographystyle{plain}
\bibliography{bib}
\end{document}